\documentclass[journal]{IEEEtran}

\ifCLASSINFOpdf
\else
\fi
\usepackage{algorithm}
\usepackage{algpseudocode}

\usepackage{amsmath}
\usepackage{cases}

\makeatletter

\newcommand{\Rmnum}[1]{\expandafter\@slowromancap\romannumeral #1@}
\makeatother

\usepackage{amsthm}  
\newtheorem{theorem}{Theorem}[]
\usepackage{amssymb}

\usepackage{graphicx} 
\usepackage{subfigure}
\usepackage{url}
 
 \usepackage{color, soul}

\begin{document}

\setulcolor{black}
\setstcolor{red}
\setstcolor{black}

\title{ Epidemic spreading in wireless sensor networks with node sleep scheduling}

\author{Yanqing Wu,~Cunlai Pu,~Gongxuan Zhang,~\IEEEmembership{Senior Member,~IEEE},~Lunbo Li,~Yongxiang Xia,~\IEEEmembership{Senior Member,~IEEE}, and ~Chengyi Xia,~\IEEEmembership{Senior Member,~IEEE}
\thanks{Yanqing Wu, Cunlai Pu, Gongxuan Zhang, and Lunbo Li are with the School of Computer Science and Engineering, Nanjing University of Science and Technology, Nanjing 210094, China (e-mail: wuyanqing, pucunlai, gongxuan, lunboli@njust.edu.cn). (Corresponding author: Cunlai Pu)}
\thanks{Yongxiang Xia is with the School of Communication Engineering, Hangzhou dianzi University, Hangzhou 310018, China (e-mail: xiayx@hdu.edu.cn).}
\thanks{Chengyi Xia is with the School of Control Science and Engineering, Tiangong University, Tianjin 300387, China (e-mail: xialooking@163.com).}
}

\maketitle

\begin{abstract}

Wireless Sensor Networks (WSNs) have become widely used in various fields like environmental monitoring, smart agriculture, and health care. However, their extensive usage also introduces significant vulnerabilities to cyber viruses. Addressing this security issue in WSNs is very challenging due to their inherent limitations in energy and bandwidth to implement real-time security measures. To tackle the virus issue, it is crucial to first understand how it spreads in WSNs. In this brief, we propose a novel epidemic spreading model for WSNs, integrating the susceptible-infected-susceptible (SIS) epidemic spreading model and node probabilistic sleep scheduling—a critical mechanism for optimizing energy efficiency. Using the microscopic Markov chain (MMC) method, we derive the spreading equations and epidemic threshold of our model. We conduct numerical simulations to validate the theoretical results and investigate the impact of key factors on epidemic spreading in WSNs. Notably, we discover that the epidemic threshold is directly proportional to the ratio of node sleeping and node activation probabilities.
 
\end{abstract}

\begin{IEEEkeywords}
Epidemic spreading, epidemic threshold, wireless sensor networks, node sleep scheduling.
\end{IEEEkeywords}

\IEEEpeerreviewmaketitle

\section{Introduction}
\IEEEPARstart{W}{ireless} Sensor Networks (WSNs) are networks consisting of small, low-power devices called wireless sensors. These sensors communicate wirelessly using technologies like Wi-Fi, Bluetooth, and Zigbee \cite{bonivento2007system,achar2022dynamics}. They collect and transmit data from the surrounding environment to a central node or base station. WSNs are widely used in domains such as environmental monitoring, industrial automation, healthcare, and surveillance for monitoring, analysis, and decision-making purposes   \cite{akyildiz2010wireless}. 

The increasing deployment of WSNs has brought about a great concern for their security  \cite{majid2022applications}. One of the prominent threats emerging in WSNs is the spread of cyber viruses and worms \cite{nwokoye2022epidemic}. These malicious entities are capable of replicating and spreading throughout sensor nodes, posing a significant risk to the confidentiality and reliability of the collected data as well as the efficiency and functionality of WSNs. Real-world instances have demonstrated the detrimental impact caused by the infection of sensor nodes \cite{bluetooth,hypponen2006malware,szor2005art}. Therefore, it becomes imperative to mitigate the spread of viruses in WSNs. However, before devising effective control mechanisms, it is crucial to understand the mechanisms underlying epidemic spreading in WSNs through comprehensive modeling and analysis.


 

 The field of network science has dedicated its efforts to understanding  epidemic spreading processes across various complex networks, such as social networks, computer networks, and biological networks \cite{fan2022epidemics,wang2020epidemic,wang2015coupled}. Numerous models have been developed by researchers to study epidemic spreading \cite{zhu2019analysis, zhang2020suppressing, wu2021traffic,ogura2020synergistic}. These models categorize nodes into distinct compartments, including susceptible, infected, and recovered states. The susceptible-infected-susceptible (SIS) and susceptible-infected-recovered (SIR) models are widely recognized as the most representative ones. In recent years, some of   these models have been applied to explain epidemic spreading in WSNs  \cite{guillen2020mathematical,liu2020novel}. For example, Wang et al.   \cite{xiaoming2009improved}   derived an iSIR model that incorporates the energy consumption of nodes for propagating worms in WSNs. Feng et al. \cite{feng2015modeling}  proposed an improved SIRS model, considering the WSN topology as uniformly distributed grid networks. Mishra et al. \cite{mishra2013mathematical}  introduced the SEIRS-V model, which additionally accounts for the exposed and vaccinated statuses of nodes to better describe worm spread dynamics in WSNs. Furthermore, Shen et al. \cite{shen2019hsird} extended the SIR model by including the dead state of nodes, indicating energy depletion in WSNs.


In WSNs, the sleep-active mode of nodes plays a crucial role in conserving energy and prolonging the network's lifetime. During the sleep state, nodes reduce their power consumption by deactivating certain components or entering a low-power mode, resulting in minimal energy usage. Conversely, during the active state, nodes are fully operational, performing tasks such as sensing, processing, and communication. This mode has a significant impact on the process of epidemic spreading within the network. For instance, when nodes are in sleep mode, they have no communication with other nodes, rendering them unable to be infected or infect others. Several works have considered node sleep scheduling in the modeling of epidemic spreading.
One example is the work by Tang et al. \cite{tang2011modified}. They proposed a modified SI epidemic spreading model in WSNs, where infected nodes can be cured during the sleep phase. Building upon this model, Tang et al.  \cite{tang2011epidemic}  further introduced two adaptive network protection schemes aimed at securing WSNs against virus attacks. In another study by Jiang et al. \cite{jiang2020virus},  the SIR virus propagation process in WSNs was examined, taking into account the influence of the duty cycle of sleep/listening.

 So far how the sleep scheduling \textcolor{black}{affects} epidemic spreading in WSNs has not been fully addressed. Furthermore, the SIS spreading dynamics in WSNs should \textcolor{black}{be} investigated, since the nodes could be infected again after recovery in the context of WSNs. The contributions of our work are summarized as follows.
\begin{itemize}
	\item {We propose a new \textcolor{black}{epidemic} spreading model in WSNs, which integrates the SIS epidemic spreading and the node  \textcolor{black}{sleep scheduling}.  Our model serves as an extension of previous models and offers valuable insights into epidemic spreading in WSNs.}
	\item  {We establish dynamic equations and drive the epidemic threshold of our spreading model by using the microscopic Markov chain (MMC) method.  The theoretical results reveal  that the epidemic threshold is determined by  the node recovery \textcolor{black}{rate}, the largest real eigenvalue of the adjacency matrix, and the node \textcolor{black}{sleep scheduling}. The theoretical findings are further confirmed by numerical simulations.  }
\end{itemize}
\par The remainder of this brief is organized as follows. In Section \Rmnum {2}, we provide a detailed description of the proposed spreading model. Section \Rmnum {3} presents the dynamic equations, equilibrium state, and epidemic threshold of the model. In Section \Rmnum {4}, we present the numerical results to explore the influence of key factors on the epidemic spreading. Finally, Section \Rmnum {5} concludes our work.

\section{Model description}
The underlying WSN consists of $N$ sensor nodes that are connected to their neighboring nodes using wireless communication techniques. In this  WSN, all sensor nodes are assumed to be stationary. If all nodes are active, i.e., there is no sleep scheduling, the network topology is fixed.  The adjacency matrix of this fixed network topology is denoted by  $\mathcal{A}=\{a_{ij}\}_{N \times N} $, where $a_{ij}=1$  indicates a connection between nodes $i$ and $j$, while  $a_{ij}=0$ signifies no connection between them.
\textcolor{black}{However, since we consider the node sleep scheduling,   the underlying network topology will be changing with time,  and  the adjacency matrix at time $t$ is denoted as $\mathcal{A}(t) = \{a_{ij}(t)\}_{N \times N}$. }


\subsection{Node sleep scheduling}
In a WSN, sensor nodes are typically resource-constrained, with limited battery power being one of the most critical constraints. Node sleep scheduling aims to mitigate this limitation by periodically putting sensor nodes into a low-power sleep mode, where they consume minimal energy. The concept behind node sleep scheduling is to dynamically activate and deactivate nodes based on their task requirements, network conditions, or specific application needs. When a sensor node is scheduled to sleep, it turns off most of its functionalities, including the radio transceiver, to minimize energy consumption. By reducing the active time of nodes, sleep scheduling effectively reduces energy dissipation and prolongs the network's operational lifetime.

In the sleep scheduling for  WSNs, at each time step, every node is  in one of the two states: active (A) and inactive or sleep (U). During each time step, all nodes independently determine their states for the next time step. A node transitions from the active state A to the sleep state  U with a probability denoted as $u$, while a node transitions from the sleep state U to the active state A with a probability denoted as $v$. Let $N_A$ and $N_U$ represent the number of nodes in the active and sleep states, respectively. Then, the node sleep scheduling can be formulated as follows:
\begin{equation}
\begin{cases}
 N_A(t+1)   = N_A(t) - u N_A(t)+v N_U(t)  , \\
 N_U(t+1)   = N_U(t) +u N_A(t) - v N_U(t)  . 
\end{cases}
\label{eq_schedule}
\end{equation}

When this scheduling  reaches the dynamic equilibrium, we have 
\begin{equation}
\begin{cases}
 N_A = \frac {v}{u+v}N  , \\
 N_U  = \frac {u}{u+v}N,    
\end{cases}
\label{eq_NA}
\end{equation}
which means the ratio of sleep  and active nodes is fixed to be   $u/v$.

\subsection{ Epidemic spreading model}
The classical SIS model is utilized  to analyze the dynamics of infection states.  In this model, each node is in one of the two states: susceptible (S) and infected (I). The infected nodes have the ability to transmit the infection to their susceptible neighbors or transition back to a susceptible state through recovery.

\begin{figure}[!]
  \centering
  \includegraphics[width=0.35\textwidth]{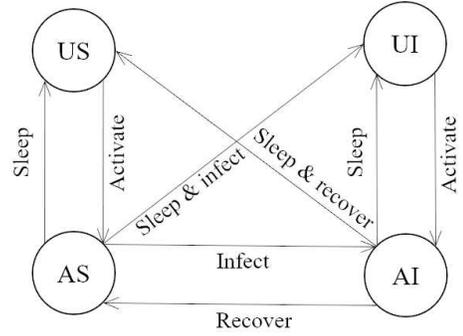}
   \vspace{-0.7em}
  \caption{ The state transition diagram of the proposed spreading  model, where US, UI, AS and AI represent  the sensor nodes are in inactive susceptible state, inactive infected state, active susceptible state and active infected state, respectively. Arrow lines indicate the transition directions between different states.}
  \label{fig_trans}
\end{figure}

Our model combines the SIS model with node sleep scheduling, resulting in four potential states: 1) US (inactive and susceptible); 2) AS (active and susceptible); 3) AI (active and infected); 4) UI (inactive and infected). It is important to note that in our model, an infected node in sleep mode cannot transmit the infection or recover from it due to a lack of energy. Additionally, we assume that at a time step,  an active susceptible node becomes infected with a probability $\beta$ when it comes into contact with an active infected neighbor, and an active infected node transitions back to a susceptible state with a probability $\gamma$. The state transition diagram of our model is depicted in Fig.~\ref{fig_trans}.

\section{Analytical  results}
In this section, we present the theoretical analysis of our proposed model. Firstly, we derive the dynamic equations governing the behavior of the model. Next, we investigate the equilibrium state of the epidemic spreading. Finally, we derive the epidemic threshold of our model.

\subsection{ Epidemic dynamics}
Let $P_i ^{US}(t)$, $P_i ^{AS}(t)$, $P_i ^{UI}(t)$ and $P_i ^{AI}(t)$ denote the probabilities that node $i$ is respectively in states US, AS, UI and AI at time step $t$. Assuming that $q_i(t)$ represents the probability that node $i$ will not be infected by any of its neighbors at time step $t$. Then,  we have
\begin{equation}
q_i(t)     = \prod_j(1 - a_{ji}(t)P_j^{AI}(t) \beta ).
\label{eq_notby}
\end{equation}

Based on the MMC method,   the dynamic equations of our model are  expressed as
\begin{equation}
\begin{cases}
\begin{aligned}
P_i^{US}(t+1)  = &P_i^{US}(t) (1 - v)  +  P_i^{AI}(t) \gamma u \\ &+ P_i^{AS}(t) q_i(t) u , 
\end{aligned} \\
\begin{aligned}
P_i^{AS}(t+1)  = &P_i^{US}(t) v + P_i^{AI}(t) \gamma (1- u) \\&+P_i^{AS}(t) q_i(t) (1-u) ,  
\end{aligned} \\
\begin{aligned}
P_i^{UI}(t+1)   = & P_i^{UI}(t) (1-v) + P_i^{AI}(t)(1- \gamma) u  \\ &+P_i^{AS}(t) (1-q_i(t)) u ,
\end{aligned} \\
\begin{aligned}
P_i^{AI}(t+1)   =& P_i^{UI}(t) v + P_i^{AI}(t)(1-  \gamma) (1- u) \\ &+P_i^{AS}(t) (1- q_i(t)) (1-u) .
\end{aligned} \\
\end{cases}
\label{eq_dynamic}
\end{equation}

 The fractions of nodes in US, AS, UI and AI states in the network at time $t$ can be respectively calculated as follows:
\begin{equation}
\begin{cases}
\rho_{US}(t)=(\sum_i^N P_i^{US}(t))/N ,  \\
\rho_{AS}(t)=(\sum_i^N P_i^{AS}(t))/N, \\
\rho_{UI}(t)=(\sum_i^N P_i^{UI}(t))/N , \\
\rho_{AI}(t)=(\sum_i^N P_i^{AI}(t))/N . \\
\end{cases}
\label{eq_frac}
\end{equation}

\subsection{Equilibrium state analysis}
As $t\to\infty$, the number of nodes in each state converges to near-constant values. Assuming that
\begin{equation}
\begin{cases}
P_i^{US}(t+1) =  P_i^{US}(t) = P_i^{US} ,  \\
P_i^{AS}(t+1) = P_i^{AS}(t) = P_i^{AS} , \\
P_i^{UI}(t+1) = P_i^{UI}(t) = P_i^{UI} , \\
P_i^{AI}(t+1) =  P_i^{AI}(t) = P_i^{AI}. \\
\end{cases}
\end{equation}
Then, Eq.~\eqref{eq_dynamic} can be written as
\begin{equation}
\begin{cases}
P_i^{US} =  P_i^{US}(1 - v)  +  P_i^{AI}\gamma u  +  P_i^{AS} q_i u , \\
P_i^{AS}  = P_i^{US} v + P_i^{AI} \gamma (1- u) +P_i^{AS} q_i (1-u) ,  \\
\begin{aligned}
P_i^{UI}  =  &P_i^{UI}(1-v) + P_i^{AI}(1- \gamma) u  \\ &+P_i^{AS} (1-q_i) u ,
\end{aligned} \\
\begin{aligned}
P_i^{AI}  =& P_i^{UI} v + P_i^{AI}(1-  \gamma) (1- u) \\ &+P_i^{AS} (1- q_i) (1-u) .
\end{aligned} \\
\end{cases}
\label{eq_dynamic2}
\end{equation}

 In our model, when the spreading process reaches equilibrium, the system can enter either of the two states: epidemic vanishment or epidemic persistence.
\subsubsection{Epidemic vanishment} $P_i^{AI} =P_i^{UI}=0$. In this state, the epidemic eventually dies out, and there are no infected nodes left in the network. Also, because $P_i^{AI} +P_i^{UI}  + P_i^{AS}  + P_i^{US}  = 1$, we have $ P_i^{US}  + P_i^{AS}  = 1$, where
\begin{equation}
\begin{cases}
P_i^{US} =  u/(u+v) , \\
P_i^{AS}  = v/(u+v)   .  \\
\end{cases}
\label{eq_vanish}
\end{equation}
We can observe that Eq. (8) is consistent with Eq.~\eqref{eq_NA}, i.e., after a finite period of time, the spreading process ends and the network runs with only active-sleep duty cycling.

\subsubsection{Epidemic persistence}   $P_i^{AI}>0$ and $P_i^{UI}>0$. In this state, the epidemic persists in the network, and the number of nodes in each state stabilizes. By solving Eq.~\eqref{eq_dynamic2}, we obtain 
\begin{equation}
\begin{cases}
P_i^{AI}  =  \frac{v}{ u+v }  \frac{1-q_i}{1-q_i+\gamma }  ,  \\
P_i^{UI}  = \frac{u}{v}  P_i^{AI} ,  \\
P_i^{AS} = \frac{v}{ u+v }  \frac{\gamma}{1-q_i+\gamma }  ,  \\
P_i^{US}  =  \frac{u}{v}   P_i^{AS} . \\
\end{cases}
\label{eq_persist2}
\end{equation}
 Eq.~\eqref{eq_persist2} is a self-consistent equation with \textcolor{black}{a trivial solution of $P_i^{AI}=P_i^{UI}=0$, where $q_i=1$. According to Eq.~\eqref{eq_persist2}, we further have  }
\begin{equation}
\begin{cases}
P_i^{US}+P_i^{UI} =  u/(u+v) , \\
P_i^{AS}+P_i^{AI}  = v/(u+v)  , \\
\end{cases}
\label{eq_persist}
\end{equation}

\subsection{Epidemic threshold}
Near the critical point of infection probability $\beta$, where the epidemic transitions from vanishing to persisting, we can assume that   $P_i^{AI}(t)=\epsilon_i \ll1$, then, $P_i^{UI}(t)=(u/v)\epsilon_i $. By substituting these assumptions into the last equation of Eq.~\eqref{eq_dynamic2}, we obtain the following expression, 
\begin{equation}
\gamma \epsilon_i =  (1- q_i)  (\frac{v}{ u+v } - \epsilon_i) .
\label{eq_lasteq2}
\end{equation}

Approximating Eq.~\eqref{eq_notby}, we can obtain the probabilities for nodes not being infected by  their neighbors as 
\begin{equation}
q_i(t)   \approx 1- \beta \sum_j a_{ji}(t)\epsilon_j .  
\label{eq_notbyapprox}
\end{equation}

Substituting Eq.~\eqref{eq_notbyapprox} into Eq.~\eqref{eq_lasteq2}, we get
\begin{equation}
\gamma \epsilon_i  -   \frac{v}{ u+v }\beta \sum_j a_{ji}(t)\epsilon_j +\beta  \epsilon_i \sum_j a_{ji}(t)\epsilon_j  =0.
\label{eq_order}
\end{equation}

By neglecting the higher-order terms in Eq.~\eqref{eq_order}, we have 
\begin{equation}
\gamma \epsilon_i  -  \frac{v}{ u+v }\beta \sum_j a_{ji}(t)\epsilon_j =0.
\label{eq_noorder}
\end{equation}

Let $e_{ji}$ be the element of the identity matrix $E$,  Eq.~\eqref{eq_noorder} can be expressed as
\begin{equation}
\sum_j^N \left(   \gamma e_{ji}  -   \frac{v}{ u+v }\beta  a_{ji}(t) \right) \epsilon_j =0.
\end{equation}
Furthmore, this equation can be written in matrix form as 
 \begin{equation}
 \left(   \gamma E  -   \frac{v}{ u+v }\beta \mathcal{A}(t) \right) \epsilon  =0 ,
 \label{eq_matrixform}
\end{equation}
where $\epsilon = (\epsilon_1, \epsilon_2, \ldots, \epsilon_N )^T$.
\begin{theorem}	
 Let $\Lambda_{max}(\mathcal{A}(t))$ and $ \Lambda_{max}(\mathcal{A})$  be the largest real eigenvalues of $\mathcal{A}(t)$ and $\mathcal{A}$, respectively, the epidemic  threshold of our proposed model  is given by
\begin{equation}
\beta_c =(1+\frac{u}{v})  \frac{\gamma}{\Lambda_{max}(\mathcal{A}(t)) } \ge (1+\frac{u}{v})  \frac{\gamma}{\Lambda_{max}(\mathcal{A}) }  .
 \label{eq_threshold}
\end{equation}
\end{theorem}

\begin{proof}
Due to space constraints, we have included the proof in the supplementary material \cite{mysupplement}, as it does not impact the comprehension of our results and conclusions.
\end{proof}


\section{Numerical results}
In this section, we present numerical results that serve to validate the theoretical derivation and discuss the effects of key factors on the epidemic threshold.
We consider a wireless sensor network with a scale-free network topology generated by the Price model. The parameters of the network topology are  number of nodes  $N = 1000$,   average node degree  $\langle k \rangle=4$, and power-law index is $\eta = 3$. At time $t = 0$, ten randomly selected nodes are infected, i.e.,  $\rho_{AI} + \rho_{UI} = 0.01$. These parameter values are kept constant unless stated otherwise.

\begin{figure}[!t]
  \centering
  \includegraphics[width=0.47\textwidth]{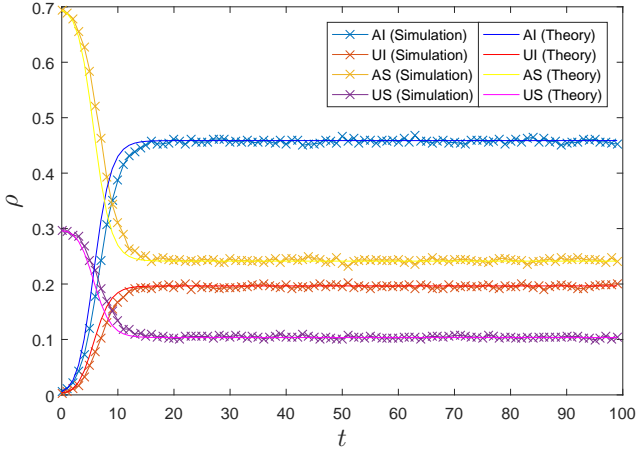}
    \vspace{-1em}
\caption{ (Color online) The temporal evolution of the fractions of nodes of different states. The parameter settings are $\beta=0.5$, $\gamma = 0.3$,  $u=0.3$, and $v = 0.7$. The solid and marked lines are the theoretical and simulation results, respectively. All results are obtained by averaging over 50 independent runs. }
 \label{fig_growth}
   \vspace{-0.7em}
\end{figure}

Figure~\ref{fig_growth} depicts the temporal evolution of the fractions of AS, US, AI, and UI nodes in our model. The results show that, under the given  parameter settings, the AI and UI curves initially experience rapid increases followed by stabilization, while the AS and US curves exhibit an initial decrease before converging. Importantly, the simulation results align well with the theoretical findings obtained using Eqs. \eqref{eq_notby} to \eqref{eq_frac}, thus emphasizing the efficacy of the MMC method in analyzing our model.

\begin{figure}[!t]
\centering 
\subfigure[]{
\includegraphics[width=0.47\textwidth]{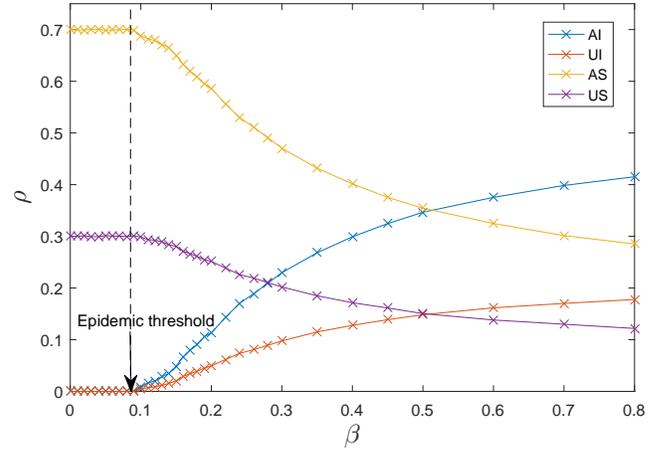}  
} 
\subfigure[]{
\includegraphics[width=0.47\textwidth]{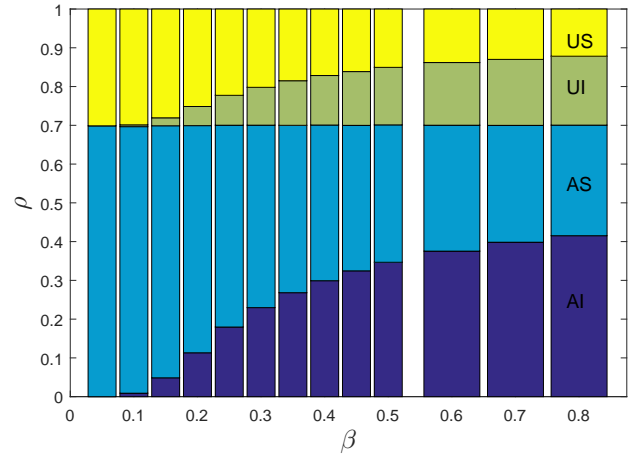}  
} \vspace{-1.2em}
\caption{(Color online) (a): The fraction of nodes of each state $\rho$ vs.  infection rate $\beta$. The black arrow line points to the epidemic threshold $\beta_c$. (b): The stacked bar chart of the fraction of nodes of each state $\rho$ for different values of $\beta$. The parameter settings  are $\gamma = 0.5$, $u = 0.3$ and $v = 0.7$. All the results are obtained by averaging over 50 independent runs.}
\label{fig_threshold}
  \vspace{-1em}
\end{figure}

Next, we present the result of the final fraction of nodes for each state, $\rho$,  as a function of the infection rate $\beta$ in Fig.~\ref{fig_threshold}. From Fig.~\ref{fig_threshold}(a), it is evident that as $\beta$ increases, there is a phase transition of $\rho$ for each of the four states. Moreover, the epidemic threshold $\beta_c$, calculated to be approximately 0.085 using Eq. (17), is further supported by the numerical result.
Furthermore, it can be observed that when $\beta < \beta_c$, only US and AS nodes remain in the network, indicating the eventual vanishing of the epidemic. The fractions of nodes in these two states agree with Eq. \eqref{eq_vanish}. On the other hand, when $\beta > \beta_c$, all states coexist, indicating the persistence of the epidemic. 
In Fig.~\ref{fig_threshold}(b), a stacked bar chart illustrates the relationship between $\rho$ and $\beta$. It is evident that the fractions of A-state nodes and U-state nodes remain constant with the increasing infection rate, \textcolor{black}{which aligns with  Eq. \eqref{eq_persist}.}

  
 \begin{figure}[!t]
  \centering
  \includegraphics[width=0.47\textwidth]{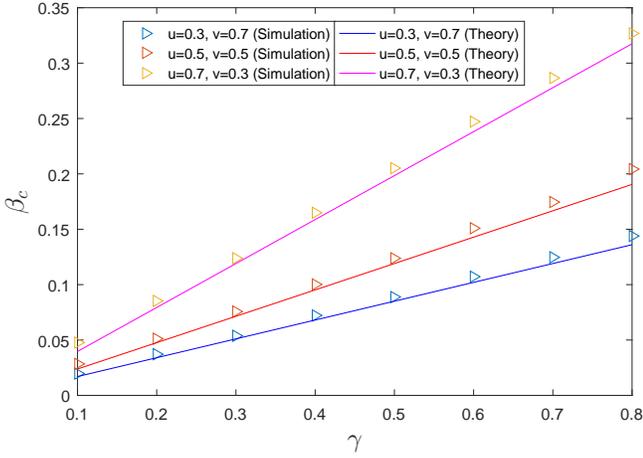}
    \vspace{-1em}
\caption{(Color online) The epidemic threshold $\beta_c$ vs. infection recovery rate $\gamma$ for different values of $u$ and $v$. The lines and  symbols are the theoretical and simulation results, respectively. All the results are obtained by averaging over 50 independent runs.}
 \label{influ_gamma}
   \vspace{-1em}
\end{figure}

 We further investigate the influence of the infection recovery rate $\gamma$ on the epidemic threshold $\beta_c$ for different settings of $u$ and $v$. The results are illustrated in Fig.~\ref{influ_gamma}. The theoretical results are derived using Eq.~\eqref{eq_threshold}.
 It can be observed that all the curves exhibit a linear increasing trend, indicating that $\beta_c$ is proportional to $\gamma$. The simulation results closely match  the theoretical predictions for $\beta_c$. These findings  indicate that a more rapid recovery process can effectively mitigate the epidemic spreading.

 \begin{figure}[!t]
  \centering
  \includegraphics[width=0.47\textwidth]{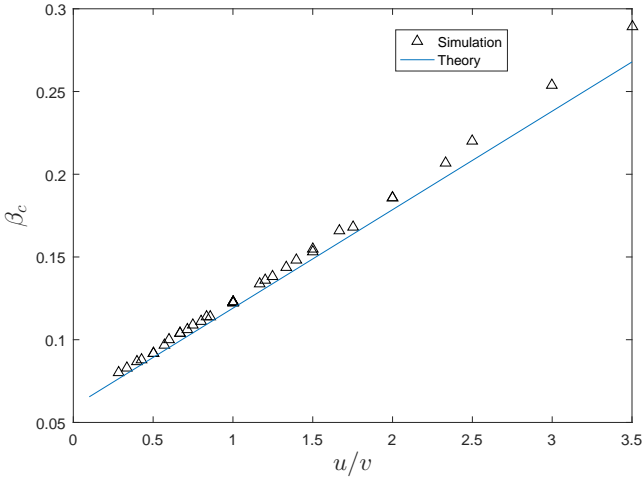}
    \vspace{-0.7em}
\caption{ (Color online) The epidemic threshold $\beta_c$ vs. the ratio   $u/v$ with $u =\{0.2:0.1:0.7\} $ and $v =\{0.2:0.1:0.7\} $. The parameter setting is  $\gamma = 0.5$. The lines and  symbols are the theoretical and simulation results, respectively. All results are obtained by averaging over 50 independent runs. }
 \label{influ_uv}
   \vspace{-0.7em}
\end{figure}

Finally, we discuss the influence of the sleep scheduling parameters $u$ and $v$ on the epidemic threshold $\beta_c$. The results are presented in Fig. \ref{influ_uv}, where the x-axis represents the ratio of $u$ to $v$, denoted as $u/v$. It is clear that $\beta_c$ exhibits a proportional relationship with $u/v$, implying a direct proportionality between $\beta_c$ and $u$, and an inverse proportionality between $\beta_c$ and $v$. Once again, the simulation and theoretical results align well.

\section{Conclusion}
In summary, our work focuses on the issue of epidemic spreading in WSNs. We present a new epidemic spreading model for WSNs that combines the SIS spreading dynamics with node sleep scheduling. Using the MMC analysis method, we derive the dynamic equations that govern the epidemic spreading process in our model. Furthermore, we analyze the conditions for the epidemic to vanish or persist and determine the outbreak threshold for the spreading.
In particular, we discover that the epidemic threshold of our model is directly proportional to the ratio of node sleeping and node activation probabilities as well as the recovery rate. Conversely, it is inversely proportional to the maximum real eigenvalue of the network's adjacency matrix.  These  theoretical findings are strongly supported by the numerical results.  
We believe that our work contributes to the understanding of epidemic spreading in WSNs and offers some insights for preventing virus attacks on these networks.


\ifCLASSOPTIONcaptionsoff
  \newpage
\fi

%
%

\end{document}